\documentclass[12pt]{article}

\usepackage{balance}
\usepackage[english]{babel}
\usepackage{amsmath,amssymb,amsthm,amsmath}

  \usepackage[T1]{fontenc}
  \usepackage{microtype}
  \usepackage{cmap} 

\usepackage{amssymb, amsthm, amsmath}
\usepackage{breqn}

\usepackage{graphicx}
\usepackage{hyperref}

\hypersetup{
bookmarksnumbered=true,
colorlinks=true,
pdfstartview={FitH},
linkcolor=blue
}

\usepackage{a4wide} 
\usepackage{fancyhdr}

\let\chapter\section
\usepackage[linesnumbered,ruled,vlined]{algorithm2e}
\SetKw{KwAssumption}{Assumption}
\usepackage{bm}
\usepackage{xcolor}
\usepackage{soul} 
\usepackage{soulutf8} 
\usepackage{enumitem}
\usepackage{mathrsfs}
\theoremstyle{theorem} 
\newtheorem{theorem}{Theorem}[section]
\newtheorem*{defi}{Definition}

\newtheorem{exam}[theorem]{Example}
\newcommand{\thistheoremname}{}
\newtheorem{genericthm}[theorem]{\thistheoremname}

\newcommand\se{\mathrm{SE}(3)}

\renewcommand\so{\mathrm{SO}(2)}
\newcommand\SO{\mathrm{SO}(3)}

\newcommand\Q{\mathbb{Q}}
\ifxetex
  \renewcommand\C{\mathbb{C}} 
\else
  \def\C{\ensuremath{\mathbb{C}}}
\fi
\newcommand\R{\mathbb{R}}
\newcommand\A{\mathbb{A}}

\newcommand{\N}{\mathbb{N}}
\renewcommand{\P}{\mathbb{P}}
\newcommand{\disc}{\operatorname{disc}}
\newcommand{\re}{\operatorname{Re}}
\newcommand{\im}{\operatorname{Im}}

\newcommand\ol[1]{\overline{#1}}
\newcommand\jc{{\bf jcapco: }}
\soulregister{\jc}{7}

\long\def\symbolfootnote[#1]#2{\begingroup%
\def\thefootnote{\fnsymbol{footnote}}\footnote[#1]{#2}\endgroup}

\makeatletter
\def\blfootnote{\xdef\@thefnmark{}\@footnotetext}
\makeindex

 \def\@testdef #1#2#3{%
   \def\reserved@a{#3}\expandafter \ifx \csname #1@#2\endcsname
  \reserved@a  \else
 \typeout{^^Jlabel #2 changed:^^J%
 \meaning\reserved@a^^J%
 \expandafter\meaning\csname #1@#2\endcsname^^J}%
 \@tempswatrue \fi}

\title{Robots, computer algebra and \\eight connected components}

\author{Jose Capco\thanks{Innsbruck University, Innsbruck, Austria, email:{\text{jose.capco@uibk.ac.at.}}},  
        \and Mohab Safey El Din\thanks{Sorbonne Universit\'e, CNRS,
    LIP6, Paris, France, email:{\text{Mohab.Safey@lip6.fr.}}}, 
  \and
  Josef Schicho\thanks{JKU University, Linz, Austria, email:{\text{Josef.Schicho@risc.jku.at}}}
}

\makeatletter
\let\mytitle\@title
\makeatother

\newtheorem{definition}{Definition}
\numberwithin{definition}{section}

\newtheorem{proposition}[definition]{Proposition}
\newtheorem{lemma}[definition]{Lemma}

\begin{document}

\maketitle

\begin{abstract}
  Answering connectivity queries in semi-algebraic sets is a long-standing and
  challenging computational issue with applications in robotics, in particular
  for the analysis of kinematic singularities. One task there is to compute the
  number of connected components of the complementary of the singularities of
  the kinematic map. Another task is to design a continuous path joining two
  given points lying in the same connected component of such a set. In this
  paper, we push forward the current capabilities of computer algebra to obtain
  computer-aided proofs of the analysis of the kinematic singularities of
  various robots used in industry.

  We first show how to combine mathematical reasoning with easy symbolic
  computations to study the kinematic singularities of an infinite family
  (depending on paramaters) modelled by the UR-series produced by the company
  ``Universal Robots''. Next, we compute roadmaps (which are curves used to
  answer connectivity queries) for this family of robots. We design an algorithm
  for ``solving'' positive dimensional polynomial system depending on
  parameters. The meaning of solving here means partitioning the parameter's
  space into semi-algebraic components over which the number of connected
  components of the semi-algebraic set defined by the input system is invariant.
  Practical experiments confirm our computer-aided proof and show that such an
  algorithm can already be used to analyze the kinematic
  singularities of the UR-series family. The number of connected components of
  the complementary of the kinematic singularities of generic robots in this
  family is $8$.
\end{abstract}

\section{Introduction}

The individual parts of a serial robot, called links, are moved by controlling
the angle of each joint connecting two links. The inverse kinematics problem
asks for the values of all angles producing a desired position of the end
effector, where ``position'' includes not just the location of the end effector
in 3-space but also the orientation. In a sense, this means inverting a function
which is called the forward kinematics map in robotics, which determines the
position of the end effector for given angles by a well-known formula (see
Section~\ref{sec:formulation}). In robot controlling, the inverse kinematics
problem is often solved incrementally: starting from some known initial angle
configuration and its corresponding end effector position, we want to compute
the change of the angles required to achieve a desired small change in the end
effector position.

Kinematic singularities are defined as critical points of the forward map, i.e. angle configurations where the Jacobian matrix
of the forward map is rank deficient. There are two known facts that make rather difficulty to control a robot in a singular
or near a singular configuration. First, if an end effector velocity or force outside the image of the singular Jacobian is desired, 
then the necessary joint velocity or torque is either not defined or very large (see \cite{mls} \S4.3 and \cite{spong} \S5.9). 
The second reason is that industrial controllers
are based on Newton's method for the incremental solution of the inverse problem, and this method is not guaranteed to converge
if it is used with a starting point close to the singular set.
For these reasons, engineers prefer to plan the robot movements avoiding kinematic singularities.

For a general serial robot with six joints, the singular set is a hypersurface
defined locally by the Jacobian determinant of the forward map. Its complementary,
a real manifold, is not connected. Counting the number of connected components
of this manifold and answering connectivity queries in this set is then of
crucial importance in this application domain.

Answering connectivity queries in semi-algebraic sets is a classical problem of
algorithmic semi-algebraic geometry which has attracted a lot of attention
through the development of the so-called ROADMAP algorithms (see e.g.
\cite{CannySAS, Canny88, BRSS, BPR00, BR, SaSc11, SaSc17, GR}). Up to our
knowledge, such algorithms had never been developed enough and implemented
efficiently to tackle real-life applications. 

\smallskip
In this paper, we push forward the capabilities of computer algebra in this
application domain by solving  connectivity queries for the non singular configuration sets of industrial robots from the UR series
of the company ``Universal Robots''. For a particular robot in this series, the UR5, the number of components of the non singular
configuration set is 8 (see Section \ref{sec:ur5}). For two points in the same component, we show how to construct a connecting path, in
two ways: either by an ad hoc way (which has its own algorithmic interest) taking advantage of the specialty of the geometric parameters of UR5, and by using the
ROADMAP algorithm (see Section \ref{sec:roadmap}).  Next, we go further and
extend our analysis of UR5 to the whole  UR series and prove that outside a
proper Zariski closed set (UR5 is outside this closed set) the number of
connected components of the non singular configuration set is constant. These
are computer-aided mathematical proofs involving <<easy>> symbolic computations 

The next contribution is based on the fact that the family of UR robots is
determined by a finite list of real parameters. Hence, an algorithmic way of
tackling the problem of analyzing kinematic singularities of the whole UR family
is to <<solve>> a {\em positive dimensional} system {\em depending on
  parameters} (i.e. after specialization of the parameters, the specialized
system is positive dimensional).   
We design an algorithm that decomposes the parameter's space into semi-algebraic subsets, such that the number of connected components of the non singular configurations
is constant in each of these subsets. As far as we know, this is the first
algorithm of that type which is designed.

We also implemented this algorithm and used it for the analysis of the kinematic
singularities of the UR series. Computations are heavy but already doable (on a
standard laptop) within $10$ hours. This is a computational way to retrieve the
same results as our computer-aided mathematical proofs. These computations show
that computer algebra today is efficient enough to solve connectivity queries
that are of practical interest in industrial robot applications.


\section{Robotics problem formulation}\label{sec:formulation}

We define a \emph{manipulator} or \emph{robot} as follows: we have finite
ordered rigid bodies called \emph{links} which are connected by $n$ revolute \emph{joints} that are also ordered. To each joint we 
associate a coordinate system or a \emph{frame}. The links are connected in 
a serial manner i.e. if we consider the robot as a graph such that the vertices are joints and the edges are links then this graph is a path (the first and 
last joint has degree 1 and all other joints have degree 2) and the joints allow rotation about its axes, so that if a joint rotates then all other 
subsequent links rotate about the axes of this joint. A reference coordinate system is chosen for the final joint which is called the 
\emph{end-effector}\footnote{this is usually another frame, but this is just an additional fixed transformation in $\se$ and w.l.o.g. we assume that the final offset, distance and twist is $0$}. 


In theoretical kinematics one may forget that the links are rigid bodies so that
\emph{collision} between links are disregarded. In this case we may as well
think of a robot as a differentiable map $F:\so^n \longrightarrow \se $ where
$\so$ is the one-dimensional group of rotations around a fixed line,
parameterised by the rotation angle, and $\se$ is the six-dimensional group of
Euclidean congruence transformations. This map is defined in the following way:
\begin{itemize}
\item The $i$-th coordinate of an element in $\so^n$ is associated to the $i$-th (revolute) joint parameter. 
\item For joint values $\vec \theta := (\theta_1,\dots,\theta_n)$ in $\so^n$,
  the image $F(\vec \theta)$ is the transformation of the end-effector from the
  initial position corresponding to all angles being zero to the final position
  obtained by composing the $n$ rotations.
\end{itemize}
The map $F$ itself is called the \emph{kinematic map} (of the robot). Its domain
is called the \emph{configuration space}, while its image is called the
\emph{work-space} or the \emph{kinematic image}.

We use the Denavit-Hartenberg (DH) convention when describing relations between
two joint frames. It is standard in robotics ; its advantages are discussed in 
e.g.\ \cite[\S3.2]{spong}, \cite[\S4.2]{angeles}. 
The transformation between the frames is given by the
following rule:
\begin{itemize}
\item The $z$-axis of the reference frame will be the axis of rotation of the joint.
\item To obtain the next frame, one starts with a rotation about the $z$-axis of the reference frame, called the \emph{rotation}, followed by
\item a translation along the $z$-axis of the reference frame, called the \emph{offset}, followed by
\item a translation along the $x$-axis, called the \emph{distance}, followed by 
\item a rotation about the $x$-axis, called the \emph{twist}.
\end{itemize}
The transformation between frame $i$ to frame $i+1$ is 
$$R_z(\theta_i)T_z(d_i)T_x(a_i)R_x(\alpha_i)$$
where $R_z,T_z,T_x,R_x$ are rotations or translations with respect to 
$z$- or $x$-axis parameterised by the angle of rotation $\theta_i$ (the $i$-th
joint parameter), the offset $d_i$, the distance $a_i$ and the angle of twist $\alpha_i$ of 
the $i$-th frame. For a given robot with $n$ joints all DH parameters except for the rotation are fixed values. So that image of $F$ for given 
joint values (the rotations) $(\theta_1,\dots,\theta_n)$ is just the multiplication of these 
transformations in $\se$. The parameters $d_1,d_n,a_n,\alpha_n$ are assumed to be $0$. This is not a loss of generality,
because we can freely choose the frame at the base and at the end-effector.
More detailed discussion on these can be seen in \cite{spong}.

\begin{exam}\label{ur5_example}
The UR5 robot has the following DH parameters:\\
{\em distances (m.)} $(a_1,\dots,a_6):=
(0,-\frac{425}{1000},-\frac{39225}{100000},0,0,0)$\\
{\em offsets (m.)}
$(d_1,\dots,d_6) :=(0,0,0,\frac{10915}{100000},\frac{9465}{100000},0)$\\
{\em twist angles (rad.)}
$(\alpha_1,\dots, \alpha_6):=(\frac{\pi}{2},0,0,\frac{\pi}{2},-\frac{\pi}{2},0)$

\noindent For example, the following joint angles (rotations, in rad.) 
$$(\theta_1,\dots,\theta_6):=\left(\frac{1}{10},\frac{2}{10},\frac{3}{10},\frac{4}{10},\frac{5}{10},\frac{6}{10}\right)$$
leads to the following transformation in $(R,\vec t) \in \se$ (represented as elements in $\SO\rtimes \R^3$) where:
{\small
\begin{align*}
R\simeq
\begin{pmatrix}
 0.047 & -0.977 & -0.209\\
-0.393 &  0.174 & -0.903\\
 0.918 &  0.123 & -0.376
\end{pmatrix},\ 
 \vec t \simeq (-6.768, -1.7784, -3.336). 
\end{align*}}
\end{exam} 

\begin{defi}
Given the kinematic map of a manipulator $F:\so^n \rightarrow \se$, the \emph{kinematic singularities} in the configuration space are the points $P\in (\P^1)^n$ such that
the Jacobian of $F$ at $P$ is rank-deficient.
\end{defi}

In this paper, we will only deal with 6-jointed manipulators. Therefore the
kinematic map is a differentiable map from the 6-dimensional configuration space
$(\so)^6$ to the group $\se$, which is also $6$-dimensional. For non-singular
points of the map, the Jacobian is therefore invertible, and $F$ is a local
homeomorphism. Here is a well-known geometric description of singularities.

\begin{theorem}\label{plueck_thm}
Let $F:\so^6 \rightarrow \se$ be the kinematic map of a robot with 6 joints. Let $P\in \so^6$. Then the following are equivalent.
\begin{enumerate}
\item $P$ is a kinematic singularity.
\item The Jacobian of $F$ at $P$ is singular.
\item If $P_1,\dots, P_6 \in \P^5(\R)$ are the Pl\"ucker representation of the
  axes (lines in $\P^3$) of the joints of the robot at the configuration point $P$ then the matrix consisting 
of the Pl\"ucker coordinates $(p_{i,j})_{i,j\le 6}$ ($P_i = (p_{i,1}:p_{i,2}:\cdots: p_{i,6})$ for $i=1,\dots,6$) is singular.
\end{enumerate}
\end{theorem}
\begin{proof}
  The equivalence of the first two items is clear by definition. The equivalence
  of the first and the third item is found in \cite[\S 4.5.1]{selig}, \cite[\S
  4.1]{mls} or \cite[\S 4.5.1.]{angeles}
\end{proof}
Assume that we have two non-singular points in the configuration set. As
explained earlier, we want to decide if these two configurations can be
connected by a curve of configurations which avoids the singular
hypersurface (see \cite{wenger} \S1.2 for some history on this question). If
yes, then an explicit construction of such a curve is also of interest. In order
to tackle these problems, we choose
parameters for $\so$ so that the equation of the hypersurface becomes a
polynomial. This is not the case when we use the angles
$\theta_1,\dots,\theta_n$, because the Jacobian contains trigonometric functions
in these angles. One well-known strategy is to parametrize by points on a unit
circle, i.e. by two parameters satisfying the equation of the unit circle. This
has a clear disadvantage: the number of variables increases, and the singular
set has co-dimension greater than one. Another well-known strategy is to replace
$\theta_i$ by $v_i=\tan\frac{\theta_i}{2}$. The variable $v_i$ ranges over
the projective line, and the angle $\pi$ corresponds to the point at infinity.
If we set $v_i=\tan\frac{\theta_i}{2}$ for $i=1,\dots,n$, then we obtain, in
general, a polynomial in $v_2,\dots,v_5$. More precisely, the degree is 2 in
$v_2$ and $v_5$ and degree 4 in $v_3$ and in $v_4$. The Jacobian does not depend
on the joint angles $\theta_1$ and $\theta_6$. This is clear from the third
characterization of singularities in Theorem~\ref{plueck_thm}: only the position
of the axes are relevant, and a rotation along the first or the last axis does
not change the position of any axis.

We define the \emph{UR Family} to be robots having a similar DH-parameter as
the known UR robots (UR5, UR10, etc.). Such  \emph{UR robots} are parameterised
by the following DH parameters \\
distances (m.)\hfill
$(a_1,\dots,a_6):= (0,a_2,a_3,0,0,0)$\\
offsets (m.)\hfill
$(d_1,\dots,d_6) :=(0,0,0,d_4,d_5,0)$\\
twist angles (rad.)\hfill $(\alpha_1,\dots,
\alpha_6):=(\frac{\pi}{2},0,0,\frac{\pi}{2},-\frac{\pi}{2},0)$

For these robots, the determinant of the Jacobian (see \cite{husty}),
expressed as a polynomial in $v_2,\dots,v_5$, is $A = -Bv_3v_5$ with 

{\small  \begin{dmath*} \label{eq:A}
B=a_2v_2^2v_3^2v_4^2-a_3v_2^2v_3^2v_4^2-2d_5v_2^2v_3^2v_4-2d_5v_2^2v_3v_4^2-2d_5v_2v_3^2v_4^2+a_2v_2^2v_3^2+a_2v_2^2v_4^2-a_2v_3^2v_4^2-a_3v_2^2v_3^2+a_3v_2^2v_4^2+4a_3v_2v_3v_4^2+a_3v_3^2v_4^2+2d_5v_2^2v_3+2d_5v_2^2v_4+2d_5v_2v_3^2+8d_5v_2v_3v_4+2d_5v_2v_4^2+2d_5v_3^2v_4+2d_5v_3v_4^2+a_2v_2^2-a_2v_3^2-a_2v_4^2+a_3v_2^2+4a_3v_2v_3+a_3v_3^2-a_3v_4^2-2d_5v_2-2d_5v_3-2d_5v_4-a_2-a_3
\end{dmath*}
} 

Note that there is a degree drop in three of the four cases: the degree in $v_3$ is
only $3$, and not $4$, and the degree in $v_4$ is only $2$, and not $4$, and the degree in $v_5$ is only
$1$, and not $2$.
The drop in the degree means that the homogeneous form of the Jacobian has a linear factor that
vanishes if and only if the value of the variable whose degree drops is infinity, or equivalently,
that the corresponding angle is $\pi$. 
Since we are interested in the complement of the singular space, we may assume
that none of these three angles is equal to $\pi$, and we can use the parameters
$v_3,v_4,v_5$ without worrying about paths crossing infinity. 

For the angle $\theta_2$, the situation is different. There is no degree drop,
hence there are configurations with $\theta_2=\pi$ in the non singular
configuration set. If we use the parametrization by half angle, then we have to
take paths in the projective line into account that cross infinity, or, in other
words, consider this variable in the projective space $\mathbb{P}^1(\R)$. But,
to take advantage on algorithms acting on semi-algebraic sets, one needs
variables that range over $\R$.

Hence, we instead use parameters $s_2=\sin(\theta_2)$ and $c_2=\cos(\theta_2)$
and add the additional equation $s_2^2+c_2^2=1$, obtaining a polynomial in the
variables $s_2,c_2,v_3,v_4,v_5$ with coefficients depending on the parameters
$a_2,a_3,d_4,d_5$. Of course, the reason for this more costly treatment for
$\theta_2$ is just necessary if we use the ROADMAP algorithm subsequently. For
an alternative analysis not using it, it is still better to use the
half tangent ranging over the projective line.



\section{Analysis of the UR5 Robot}\label{sec:ur5}

We gave in Example \ref{ur5_example} the Denavit-Hartenberg parameters of the
UR5 robot.
These values are used to instantiate $a_2, a_3$ and $d_5$ in the above
polynomial $B$; the specialized polynomial is then denoted by $\tilde{B}$ and we
let $\tilde{A} = \tilde{B}v_3v_5$.
Recall that $v_2$ ranges over $\P^1$, while $v_3,v_4,v_5$ range only over the affine line.

We investigate the discriminant of $\tilde{B}$ with respect to the variable $v_2$ (thus
the projection of the critical set to the $(v_3,v_4)$-plane). The discriminant
of $\tilde{B}$ with respect to the variable $v_2$ we denote as $b\in \R[v_3,v_4]$. This
discriminant is still factorisable in $\C[v_3,v_4]$. In fact, one checks, that
it is the factor of two complex conjugates of some polynomial in $\R[v_3,v_4]$.
This implies that $b=c^2+d^2$ is the sum of two squares of real polynomials
$c,d\in \R[v_3,v_4]$. These two polynomials are given by
{\small\begin{align*}
c&=\frac{1577212v_3-3561263v_4-14850585}{\sqrt{2006237}}\\
d&=\frac{(\sqrt{2006237}v_4+1239915-7144712)v_3+16090500v_4}{\sqrt{2006237}}
\end{align*}}
Thus, $b$ can have only two real roots (i.e. two pairs $(v_3,v_4)$), i.e the vanishing set of $b$ in $\R^2$ is finite, namely they points that are the zeros of both $c$ and $d$. 
We solve this as 
floating numbers to have an idea of their vicinity in an affine chart of the ambient space of the kinematic singularity. The roots are 
$$
\begin{array}{llll}
q_1&=(v_3 \simeq -9.140975564&,& v_4 \simeq -8.218388067) \\ 
q_2&=(v_3 \simeq 9.140975563&,& v_4 \simeq -.1216783622)
\end{array}
$$

For the two special values $q_1$ and $q_2$ in the $(v3,v4)$-plane, all three coefficients of $\tilde{B}$ with respect to $v_2$ are zero.

Now, since the discriminant $b$ is positive except at these two points and since $\tilde{B}$ itself is quadratic with respect to $v_2$ we conclude that the preimage of the projection (to 
$(v_3,v_4)$-plane) are two real points in the variety defined by $\tilde{B}$. Thus, the variety defined by $\tilde{B}$ is composed of two sheets (above any two points $(v_3,v_4)$ 
except $q_1$ and $q_2$). Let $X$ be the complement of the vanishing points of
$\tilde{B}$ in $\P^1\times \A^2$.  Set
$$Y:=\A^2\setminus (\{q_1,q_2\}\cup\{ (0,v_4) \mid v_4\in\R\})$$
So we have a canonical projection (to the $(v_3,v_4)$-plane) from $X\cap (\P^1\times Y)$ to $Y$. The fiber of this projection is a projective line without two distinct points. 
Hence, every fiber has two components. The sign of $\tilde{B}$ is different for the two components of each fiber. Then, we have two components of $X$ for each component of $Y$. 
Obviously, $Y$ has two components, hence we have a total number of $4$ components. 

For the non singular set, which is the complement of the zero set of $\tilde{A}$, we get $8$ components: for each component of $X$, we have one component where $v_5$ is positive and one 
where $v_5$ is negative. 

Now assume that we have two non singular configuration points $x,y$ in the same component, and we want to construct a path connecting them. The projections to $Y$ have to lie in 
the same component of $Y$, and because $Y$ is the plane without a line and two points, it is easy to connect the images of the projections in $Y$: in most cases, a straight line 
segment is fine; if the straight line segment connecting the two image points contains $q_1$ or $q_2$, we have to do a random detour via a third point. The zero set of $\tilde{B}$ is a 
two-sheeted covering of $Y$. So, for any value of $Y$, we have two points in the zero set of $\tilde{B}$ projecting to it. If we look at these points as points in $\so$, then it is clear 
that there are 
two ``midpoints'' in the zero set of $\tilde{B}$, which have equal angle distance to these two points. The value of $\tilde{B}$ is positive for one of the two midpoints and negative for the 
other one. The 
sign of $\tilde{B}(x)$ and $\tilde{B}(y)$, however, must be the same because the two points are in the same connected component. Suppose, without loss of generality, that $\tilde{B}(x)$ and $\tilde{B}(y)$ are 
both 
positive. Then we first connect $x$ to the midpoint over the projection of $x$ to $Y$ with positive sign, by a curve in the fiber. Next, we lift the path in $Y$, 
connecting the projections of $x$ and $y$ in the same component, to a path of 
midpoints with positive sign, arriving at the midpoint with positive sign lying over the projection of $y$. Finally, we connect this midpoint to $y$ by the other fiber.

\noindent Below, we show the sheets in Figure~\ref{fig:twosheets} to illustrate
that:

\begin{enumerate}[wide, labelwidth=!, labelindent=0pt,label={(\roman*)}]
\item the regions above and below the sheets can be connected
\item the region between the two sheets is the other connected component
\item the two points $q_2$ and $q_3$ are points in the projection where the sheets get connected (see assymptotes in Fig.\ \ref{fig:twosheets}). Thus, the variety describing the two sheets is connected.
\end{enumerate}
\vspace{-0.5cm}

\begin{figure}[!htbp]
  \resizebox{6cm}{!}{\Large
\setlength{\unitlength}{524.11450195bp}%
\begin{picture}(1,0.5423947)
\put(0,0){\includegraphics[width=\unitlength]{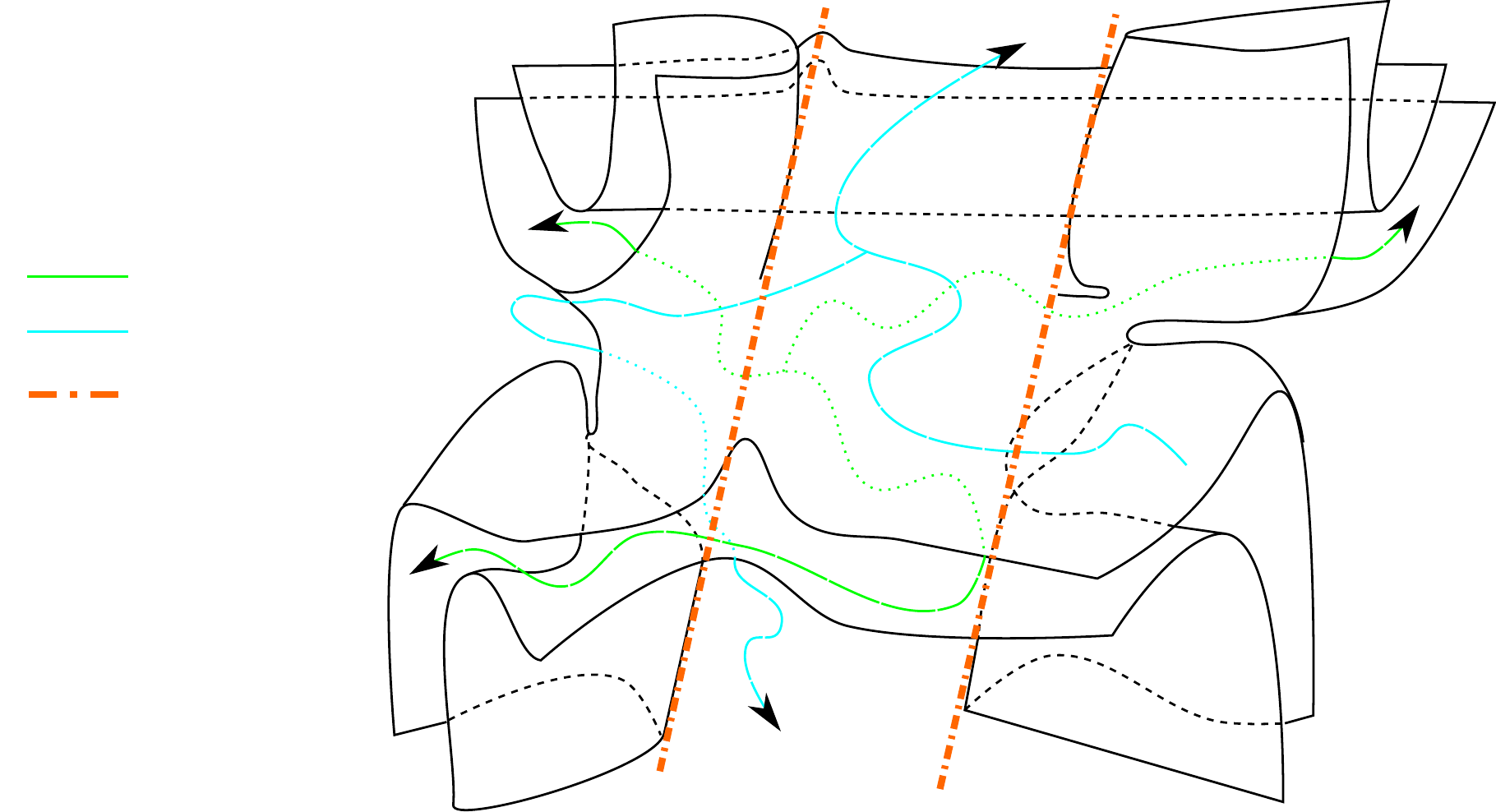}}
\put(0.01557385,0.40716984){\bf Legends:}
\put(0.09758084,0.35303333){path in component 1}
\put(0.09758085,0.31606529){path in component 2}
\put(0.09910724,0.27411097){asymptotes}
\end{picture}
  }
\caption{\small The two sheets of $\tilde{B}$}\label{fig:twosheets}
\end{figure}
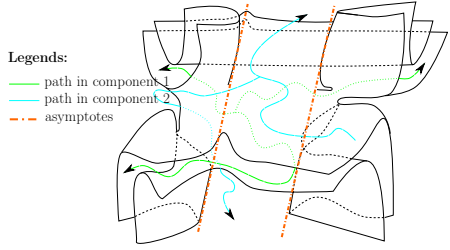


\section{UR series}\label{ur_series}

We can make a general statement for robots belonging to the UR family (e.g.\ UR10, UR3 etc.\ ). We define the \emph{UR Family} to be robots which have a similar DH-parameter as the known UR robots (UR5, UR10 etc.), a robot in 
this family we shall call a \emph{UR robot}. Namely they are parameterised by
the following DH parameters :\\
distances (m.)\hfill
$(a_1,\dots,a_6):= (0,a_2,a_3,0,0,0)$\\
offsets (m.)\hfill
$(d_1,\dots,d_6) :=(0,0,0,d_4,d_5,0)$\\
twist angles (rad.)\hfill
$(\alpha_1,\dots, \alpha_6):=(\frac{\pi}{2},0,0,\frac{\pi}{2},-\frac{\pi}{2},0)$

\noindent i.e. these robots are parameterised by $4$ parameters: $a_2,a_3,d_4,d_5$.
 
We can write the largest (in number of terms and in degree) polynomial factor of the polynomial whose vanishing points is the kinematic singularity in configuration space of 
a UR robot as
{\small\begin{dmath*}
B=a_2v_2^2v_3^2v_4^2-a_3v_2^2v_3^2v_4^2-2d_5v_2^2v_3^2v_4-2d_5v_2^2v_3v_4^2-2d_5v_2v_3^2v_4^2+a_2v_2^2v_3^2+a_2v_2^2v_4^2-a_2v_3^2v_4^2-a_3v_2^2v_3^2+a_3v_2^2v_4^2+4a_3v_2v_3v_4^2+a_3v_3^2v_4^2+2d_5v_2^2v_3+2d_5v_2^2v_4+2d_5v_2v_3^2+8d_5v_2v_3v_4+2d_5v_2v_4^2+2d_5v_3^2v_4+2d_5v_3v_4^2+a_2v_2^2-a_2v_3^2-a_2v_4^2+a_3v_2^2+4a_3v_2v_3+a_3v_3^2-a_3v_4^2-2d_5v_2-2d_5v_3-2d_5v_4-a_2-a_3
\end{dmath*}
}
Note that $d_4$ does not affect the singularity of the robot. Taking the
discriminant of $B$ with respect to $v_2$ yields the sum of two squares i.e. the
product of two quadratic complex conjugate polynomials $\disc(B,v_2) = g\ol g$.
\begin{dmath*}
g = (-a_2v_3v_4+a_3v_3v_4+d_5v_3+d_5v_4+a_2+a_3) + (-d_5v_3v_4+a_2v_3+a_2v_4-a_3v_3+a_3v_4+d_5)i
\end{dmath*}
For a robot determined by some real quadruple $u\in\R^4$, let $A_u,B_u,g_u$ be the polynomials obtained by instantiating in $A,B,g$
the variables $a_2,a_3,d_4,d_5$ by the corresponding real values in the quadruple. 
Let $f:\R^3\to\R^2$ be the projection $(v_2,v_3,v_5)\mapsto (v_3,v_5)$. 
Let $Y_u\subset\R^2$ be the complement of (the union of the line $v_3=0$ and the common zero set of $\re g_u$ and $\im g_u)$.
Then the real zero set of $B_u$ in $\R^3$ intersected with $f^{-1}(Y_u)$ projects to surjectively to $Y_u$, in such a
way that there are two sheets, each projecting homeomorphically to $Y_u$.

For general robot $u$, the real set 
of $g$, which is meaning the set of all points in the real $(v_3,v_4)$-plane such that both the real part and the imaginary 
part of $g$ is equal to zero, is a finite subset of $\R^2$. 
All arguments from the previous sections work in this case as well. Hence 
we get $8$ components for these parameters'values. Moreover, we have paths connecting points in the same component, as in the previous section.

It remains to treat the non-general robots where the real zero set of $g$ is one-dimensional. This is the case if and only if 
$d_5=a_2^2-a_3^2=0$. The even more special case $d_5=a_2=a_3=0$ is easy to analyze: here, the determinant of the Jacobian $A$ is identically zero, which means that there are no 
non singular configurations. Excluding that case, we have two families of robots, and in each family, up to the value of $d_4$, the parameters are unique up to scaling. Without loss of generality, we can reduce to exactly two non-general robots $u'=(1,1,0,0)$ and $u''=(1,-1,0,0)$. 
Then the polynomial $B_{u'}$ has a factor is $C':=v_3v_4-v_3-v_4-1$, and the polynomial $B_{u''}$ has a factor
$C'':=v_3v_4+v_3+v_4-1$.
Apart from that complication, the analysis proceeds similar as in the general case: the set $Y_{u'}$ is the plane minus 
the line $v_3=0$ minus the hyperbola with equation $C'$, and the set $Y_{u''}$ is the plane minus the hyperbola with equation $C''$.
In both cases, the number of components of $Y$ is 5, as it can be
seen in Figure~\ref{fig:specialcase}.
Consequently, we have 20 components in total. The paths between points in the same component can be constructed 
similarly as in the general case.
\begin{figure}[b]\centering
\includegraphics[scale=0.2]{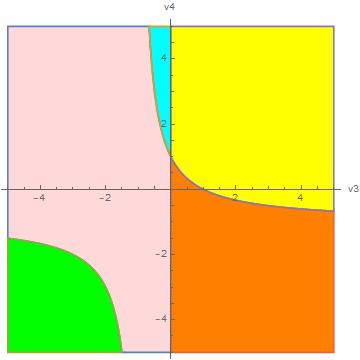}
\includegraphics[scale=0.2]{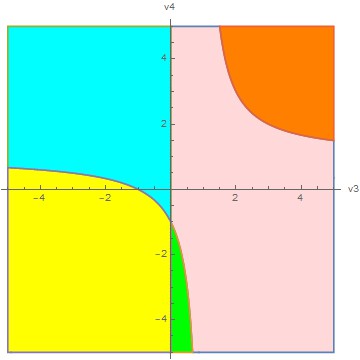}
\caption[.]{\small Left (resp. right) shows the components of $v_3(v_3v_4+v_3+v_4-1)\ne
  0$ (resp. $v_3(v_3v_4-v_3-v_4-1)\ne 0$) in $Y$}
\label{fig:specialcase}
\end{figure}


\section{Connectivity and roadmaps}\label{sec:roadmap}

We explain the ROADMAP algorithm for the special case where the semi-algebraic
set $S$ is given as a subset of some vector space $\R^N$, $N\in\N$, defined by
an equation $f(x_1,\dots,x_N)=0$ and an inequation $g(x_1,\dots,x_N)\ne 0$. We
assume that the algebraic set defined by $f=0$ is smooth. This is sufficient for
our application: the inequation is the determinant of the Jacobian of the
kinematic map $A$, and the equality is $s_2^2+c_2^2-1=0$.

One first reduces the problem to one where the semi-algebraic set we
consider is bounded. Note that there exists $R>0$ large enough such
that the connected components of $S$ are in one-to-one correspondence with the
intersection of $S$ with the hyper-ball defined by $\mathscr{N}_R\leq 0$ where
$\mathscr{N}_R = x_1^2+\cdots+x_n^2- R$. We denote this intersection by $S'$.
Note that a roadmap of $S'$ provides a roadmap of $S$. 

Determining such a large enough real number $R$ is done by choosing it larger
than the largest critical value of the restriction of the map $\bm{x}\to
\|\bm{x}\|^2$ to each regular strata of the the Euclidean closure of $S$. This
leads us to compute critical values of that map restricted to the hypersurface
defined by $f=0$ and next take the limits of the critical values of the sets
defined by $g=\pm\varepsilon$ and $f=g\pm\varepsilon=0$ when $\varepsilon\to 0$.

Next, we compute the critical values $\eta_1 < \cdots < \eta_s$ of the
restriction of the map $\bm{x}\to g(\bm{x})$ to the semi-algebraic set defined
by $f=0$ and $\mathscr{N}_R\leq 0$. Following Thom's isotopy lemma \cite{CS95}, when
$e$ is chosen between $0$ and $\min(|\eta_i|, 1\leq i \leq s)$, the connected
components of the semi-algebraic set $S^+_e$ (resp. $S^-_e$) defined by
$\mathscr{N}_R\leq 0, f=g-e=0$ (resp. $\mathscr{N}_R\leq 0, f=g+e=0$) are in
one-to-one correspondence with the connected components of the semi-algebraic
set defined by $\mathscr{N}_R\leq 0, f=0, g>0$ (resp. $\mathscr{N}_R\leq 0, f=0,
g>0$). Besides, $S^+_e\subset S'$ (resp. $S^-_e\subset S'$). Then a roadmap of
$S'$ is obtained by taking the union of a roadmap of $S^+_e$ with the roadmap of
$S_e^-$. Hence, we have performed a reduction to computing roadmaps in the
compact semi-algebraic sets $S_e^+$ and $S_e^-$.

In our application, the algebraic sets defined by the vanishing of all subsets
of the defining polynomials of $S_e^+$ and $S_e^-$ are smooth. Hence, we can
rely on a slight modification of the roadmap algorithm given in \cite{CannySAS}
where we replace computations with multivariate resultants for solving
polynomial systems by computations of Gr\"obner bases.

The algorithm in \cite{CannySAS} then takes as input a polynomial system
defining a closed and bounded semi-algebraic set $S$ and proceeds as follows.
The core idea is to start by computing a curve $\mathscr{C}$ which has a
non-empty intersection with each connected component of $S$. That curve will be
typically the critical locus on the $(x_1, x_2)$-plane when one is in generic
coordinates (else, one just needs to change linearly generically the coordinate
system). A few remarks are in order here. When $S$ is defined by
$f_1=\cdots=f_p=0$ and $g_1\geq 0,\cdots, g_s\geq 0$, to define the critical
locus of the projection on the $(x_1, x_2)$-plane restricted to $S$ one takes
the union of the critical loci of that projection restricted to the real
algebraic sets defined for all $\{i_1, \ldots, i_\ell\}\subset \{1, \ldots,
s\}$, by $f_1=\cdots=f_p=g_{i_1}=\cdots=g_{i_\ell}=0$ and intersect this union
of critical loci with $S$ (see~\cite{CannySAS}).

That way, one obtains curves that intersect all connected components of $S$ but
these intersections may not be connected. To repair these connectivity failures,
Canny's algorithm finds appropriate slices of $S$. Let $\pi_1$ be the canonical
projection $(x_1, \ldots, x_n)\to x_1$. This basically consists in
finding $\alpha_1 < \ldots <\alpha_k$ in $\R$ such that the union of $\cup_{i=1}^kS\cap
\pi_1^{-1}(\alpha_i)$ with the critical curve $\mathscr{C}$ has a non-empty and
connected intersection with each connected component of $S$. 

The way Canny proposes to find those $\alpha_i$'s is to compute the critical
values of the restriction of $\pi_1$ to $\mathscr{C}$. By the algebraic Sard's
theorem (see e.g. \cite[Appendix B]{SaSc17}), these values are in finite number
and Canny proposes to take $\alpha_1, \cdots, \alpha_k$ as those critical
values. This leads to compute with real algebraic numbers which can be encoded
with their minimal polynomials and isolating intervals. Since these minimal
polynomials may have large degrees (singly exponential in $n$), that step can be
prohibitive for practical computations. We use then the technique introduced in
\cite{MS04} which consists in replacing $\alpha_1 < \cdots < \alpha_k$ with
rational numbers $\rho_1 < \cdots < \rho_{k-1}$ with $\alpha_i < \rho_i <
\alpha_{i+1}$. We refer to \cite{MS04} for the rationale justifying this trick.
All in all, one obtains a recursive algorithm with a decreasing number of
variables at each recursive call. Combined with efficient Gr\"obner bases
engines, we illustrate in Section~\ref{sec:experiments} that the ROADMAP
algorithm (with the modifications introduced above) can be used in practice to
answer connectivity queries in semi-algebraic sets in concrete applications.

The concept of roadmap and the algorithm computing it, described above, may seem
cumbersome and unnecessarily sophisticated, especially when compared with the
much more direct CAD approach \cite{ScSh}. The CAD algorithm is
also a recursive algorithm, producing its recursive instance by projecting the
hypersurface to $\R^{n+1}$ and analyzing the discriminant. This leads to an
iteration of discriminants, and it is easy to see that the degree of the
iterated discriminants grows double exponentially in $n$: roughly, the degree of
the discriminant is squared in every iteration. There lies the motivation for
all the sophistication of the ROADMAP algorithm: for each instance in the all
recursive calls, the degree of the input polynomial is exactly the same as the
degree of the initially given polynomial $f$. This leads to an asymptotic
complexity which is only single exponential in $n^2$. We refer to \cite{SaSc11,
  BRSS, SaSc17, BR} for
more recent algorithms improving the complexity of roadmap computations.



\section{Parametric polynomial systems}

Let $F = (f_1, \ldots, f_p)$ and $G=(g_1, \ldots, g_q)$ in $\Q[\bm{x}, \bm{y}]$
with $\bm{x} = (x_1, \ldots, x_n)$ and $\bm{y} = (y_1, \ldots, y_t)$. We
consider further $\bm{y}$ as a sequence of parameters and the polynomial system
\[
f_1=\cdots=f_p=0, \qquad g_1\ \sigma_1\  0,\; \ldots, \;g_q\ \sigma_q\ 0
\]
with $\sigma_i\in \{>, \geq\}$. We let $S\subset \R^n\times \R^t$ be the
semi-algebraic set defined by this system. For $y\in \R^t$, we denote by $F_y$
and $G_y$ the sequences of polynomials obtained after instantiating $\bm{y}$ to
$y$ in $F$ and $G$ respectively. Also, we denote by $S_y\subset \R^n$ the
semi-algebraic set defined by the above system when $\bm{y}$ is specialized to
$y$. The algebraic set defined by the simultaneous vanishing of the entries of
$F$ (resp. $F_y$) is denoted by $V(F)\subset \C^{n+t}$ (resp. $V(F_y)\subset
\C^n$). 

We describe an algorithm for solving such a parametric polynomial system {\em
  without} assuming that for a {\em generic} point $y$ in $\C^t$, $V(F_y)$ is
finite. In that situation, solving such a parametric polynomial system may
consist in partitioning the parameters'space $\R^t$ into semi-algebraic sets
$T_1, \ldots, T_r$ such that, for $1\leq i \leq r$, the number of connected
components of $S_y$ is invariant for any choice of $y$ in $T_i$. We 
prove below that such an algorithmic problem makes sense.

\begin{proposition}
  Let $S\subset \R^n\times \R^t$ be a semi-algebraic set and $\pi$ be the
  canonical projection
\[
  (x_1, \ldots, x_n, y_1, \ldots, y_t) \to (y_1, \ldots, y_t).
\]
There exist semi-algebraic sets $T_1, \ldots, T_r$ in $\R^t$ such that
\begin{itemize}
\item $\R^t = T_1\cup\cdots \cup T_r$,
\item there exists $b_i\in \N$ such that for any $y\in T_i$, the number of
  connected components of $S_y$ is $b_i$.
\end{itemize}
\end{proposition}

\begin{proof}
  Observe that the restriction of $\pi$ to $S$ is semi-algebraically continuous.
  From Hardt's semi-algebraic triviality theorem \cite[Theorem 9.3.2]{bcr},
  there exists a finite partition of $\R^t$ into semi-algebraic sets $T_1,
  \ldots, T_r$ and for each $1\leq i \leq r$, a trivialization $\vartheta_i:
  T_i\times E_i \to \pi^{-1}(T_i)\cap S$ (where $E_i$ is a fiber
  $\pi^{-1}(y)\cap S$ for some $y\in T_i$).
  Fix $i$ and choose an arbitrary point $y'\in T_i$. Observe that we are done
  once we have proved that $\pi^{-1}(y')\cap S$ and $E_i$ have the same number
  of connected components. Recall that, by definition of a trivialization (see
  \cite[Definition 9.3.1]{bcr}), $\theta_i: T_i\times E_i \to \pi^{-1}(T_i)\cap
  S$ is a semi-algebraic homeomorphism and for any $(y', x)\in T_i\times E_i$,
  $\pi\circ\theta_i (y', x) = y'$. Hence, we deduce that $E_i$ is homeomorphic
  $\pi^{-1}(y')\cap S$. As a consequence, they both have the same number of
  connected components.
\end{proof}

Instead of computing a partition of the parameters'space into semi-algebraic
sets $T_1, \ldots, T_r$ as above, one will consider non-empty {\em disjoint}
open semi-algebraic sets $U_1, \ldots, U_\ell$ in $\R^t$ such that the
complement of $U_1\cup\cdots\cup U_\ell$ in $\R^t$ is a semi-algebraic set of
dimension less than $t$ and such that for $1\leq i \leq t$, there exists
$b_i\in\N$ such that $b_i$ is the number of connected components of $S_y$ for
any $y\in U_i$. For instance, one can take $U_1, \ldots, U_\ell$ as the
non-empty interiors (for the Euclidean topology) of $T_1, \ldots, T_r$.

Our strategy to solve this problem is to first compute a polynomial $\Delta$ in
$\Q[\bm{y}]-\{0\}$ defining a Zariski closed set $\mathscr{D}\subset \C^t$ such
that $\mathscr{D}$ contains $\R^t-\left ( U_1\cup\cdots\cup U_\ell\right )$. The
next lemma is immediate.

\begin{lemma}
  Let $\mathscr{E}\subset \R^t$ be a finite set of points which has a non-empty
  intersection with any of the connected components of the semi-algebraic set
  defined by $\Delta\neq 0$. For $1\leq i \leq \ell$, $\mathscr{E}\cap
  U_i$ is not empty.
\end{lemma}

Hence, computing sample points in each connected component of the set defined by $\Delta \neq 0$ (e.g. using the algorithm in \cite{SaSc03}
applied to the set defined by $z\Delta - 1 = 0$ where $z$ is a new
variable) is enough to obtain at least one point per connected component of
$U_1, \ldots, U_\ell$. Finally, for each such a point $y$, it remains to
count the number of connected components of the set $S_y$ by
using a roadmap algorithm.

We call {\em partial semi-algebraic resolution} of $(F, G)$ the data $(b_1,
\eta_1), \ldots, (b_k, \eta_k)$ where $b_i$ is the number of connected
components of $S_{\eta_i}$ and $\{\eta_1, \ldots, \eta_k\}$ has a non-empty
intersection with each connected component of $U_1\cup\cdots\cup U_{\ell}$.

Hence, our algorithm relies on three subroutines. The first one, which we call
{\sf Eliminate}, takes as input $F$ and $G$, as well as $\bm{x}$ and $\bm{y}$
and outputs $\Delta\in \Q[\bm{y}]$ as above ; we let $\mathscr{D} = V(\Delta)$.
The second one, which we call {\sf SamplePoints} takes as input $\Delta$ and
outputs a finite set of sample points $\{\eta_1, \ldots, \eta_k\}$ (with
$\eta_i\in\Q^t$) which meets each connected component of $\R^t - \mathscr{D}$.
The last one, which we call {\sf NumberOfConnectedComponents} takes $F_\eta$ and
$G_\eta$ and for some $\eta \in \Q^t$ and computes the number of connected
components of the semi-algebraic set $S_\eta$. The algorithm is described
hereafter.

\begin{center}
\begin{algorithm}
\KwData{Finite sequences $F$ and $G$ in $\Q[\bm{x}, \bm{y}]$ with $\bm{x} =
  (x_1, \ldots, x_n)$ and $\bm{y} = (y_1, \ldots, y_t)$.}
\KwResult{a partial semi-algebraic resolution of $(F, G)$}
\SetAlgoNoLine
{$\Delta \gets {\sf Eliminate}(F, G, \bm{x}, \bm{y})$}\\
{$\{\eta_1, \ldots, \eta_k\} \gets {\sf SamplePoints}(\Delta\neq 0)$}\\
\For{$i$ from $1$ to $k$}{
  $b_i = {\sf NumberOfConnectedComponents}(F_{\eta_i}, G_{\eta_i})$
}
\Return $\{(b_1, \eta_1), \ldots, (b_k, \eta_k)\}$.
\caption{{\sf ParametricSolve$(F, G, \bm{x}, \bm{y})$}}
\label{algo:solveparam}
\end{algorithm}
\end{center}
While the rationale of algorithm {\sf ParametricSolve} is mostly
straightforward, detailing each of its subroutines is less. The easiest ones are
{\sf SamplePoints} and {\sf NumberOfConnectedComponents}: they rely on
known algorithms using the critical point method \cite{Basu97, BPR}, polar
varieties \cite{SaSc03, S05, BGHSS, BGHM} and for computing roadmaps
\cite{BPR00, BRSS, SaSc17, SaSc11}. 

The most difficult one is subroutine {\sf Eliminate}. We provide a detailed
description of it under the following regularity assumption. We say that $(F, G)$
satisfies assumption $({\sf A})$ 
\begin{itemize}
\item[$({\sf A})$] for any $\{i_1, \ldots, i_s\}$ in $\{1, \ldots, q\}$, the
  Jacobian matrix associated to $(f_1, \ldots, f_p, g_{i_1}, \ldots, g_{i_s})$
  has maximal rank at any complex solution to
  \[
    f_1=\cdots=f_p=g_{i_1}=\cdots=g_{i_s}=0
  \]
\end{itemize}
Note that using the Jacobian criterion \cite[Chap. 16]{Eisenbud}, it is easy to decide whether $({\sf A})$
holds. Note also that it holds generically.

For $\bm{i} = \{i_1, \ldots, i_s\}\subset \{1, \ldots, q\}$, under assumption
$({\sf A})$, the algebraic set $V_{\bm{i}}\subset \C^{n+t}$ defined by
  \[
    f_1=\cdots=f_p=g_{i_1}=\cdots=g_{i_s}=0.
  \]
  are smooth and equidimensional and these systems generate radical ideals
  (applying the Jacobian criterion \cite[Theorem 16.19]{Eisenbud}). Besides, the
  tangent space to $z\in V_{\bm{i}}$ coincides with the the (left) kernel of
  the Jacobian matrices associated to $(f_1, \ldots, f_p, g_{i_1}, \ldots,
  g_{i_s})$ at $z$.

  Let $I$ be the ideal generated by $(f_1, \ldots, f_p, g_{i_1}, \ldots,
  g_{i_s})$ and the maximal minors of the truncated Jacobian matrix associated
  to $(f_1, \ldots, f_p, g_{i_1}, \ldots, g_{i_s})$ obtained by removing the
  columns corresponding to the partial derivatives w.r.t. the
  $\bm{y}$-variables. Under assumption $({\sf A})$, one can compute the set of
  critical values of the restriction of the projection $\pi$ to the algebraic
  set $V_{\bm{i}}$ by eliminating the variables $\bm{x}$ from $I$.

  Hence, using elimination algorithms, which include Gr\"obner bases \cite{F4,
    F5} with elimination monomial orderings, or triangular sets (see e.g.
  \cite{Wang01, ALM99}) or geometric resolution algorithms \cite{GeoReso,
    GiHeMoMoPa98, GiHeMoPa95}, one can compute a polynomial $\Delta_{\bm{i}}\in
  \Q[\bm{y}]$ whose vanishing set is the set of critical values of the
  restriction of $\pi$ to $V_{\bm{i}}$. By the algebraic Sard's theorem (see e.g.
  \cite[App. A]{SaSc17}), $\Delta_{\bm{i}}$ is not identically zero (the 
  critical values are contained in a Zariski closed subset of $\C^t$).

  Under assumption $({\sf A})$, we define the set of critical points
  (resp. values) of the restriction of $\pi$ to the Euclidean closure of $S$ as
  the union of the set of critical points (resp. values) of the restriction of
  $\pi$ to $V_{\bm{i}}\cap \R^{n+t}$ when $\bm{i}$ ranges over the subsets of
  $\{1, \ldots, q\}$. We denote the Euclidean closure of $S$ by $\overline{S}$,
  the set of critical points (resp. values) of the restriction of $\pi$ to
  $\overline{S}$ by $\mathscr{W}(\pi, \overline{S})$ (resp. $\mathscr{D}(\pi,
  \overline{S})$).

  We say that $S$ satisfies a {\em properness} assumption ${\sf (P)}$ if:
  \begin{itemize}
  \item[${\sf (P)}$] the restriction of $\pi$ to $\overline{S}$ is proper
    ($\forall y\in \R^t$, there exists a ball $B\ni y$ s.t.
    $\pi^{-1}(B)\cap\overline{S}$ is closed and bounded).
  \end{itemize}

  Our interest in critical points and values is motivated by the semi-algebraic
  version of Thom's isotopy lemma (see \cite{CS95}) which states the following,
  under  assumption ${\sf (P)}$. Take an open semi-algebraic subset
  $U\subset \R^t$ which does not meet the set of critical values of the
  restriction of $\pi$ to $\overline{S}$, $y\in U$ and $E = \pi^{-1}(y)\cap S$.
  Then, there exists a semi-algebraic trivialization $\vartheta: U \times E\to
  \pi^{-1}(U)\cap S$. 

  Hence, $\cup_{\bm{i}\subset \{1, \ldots, q\}}\mathscr{D}(\pi, V_{\bm{i}})$
  contains the boun\-daries of the open disjoint semi-algebraic
  set $U_1, \ldots, U_\ell$. Recall that by Sard's theorem it has
  co-dimension $\geq 1$. This leads to the following
  algorithm.
\begin{center}
\begin{algorithm}
\KwData{Finite sequences $F$ and $G$ in $\Q[\bm{x}, \bm{y}]$ with $\bm{x} =
  (x_1, \ldots, x_n)$ and $\bm{y} = (y_1, \ldots, y_t)$, defining a
  semi-algebraic set $S\subset \R^n \times \R^t$. \\
Assumes that  assumptions $({\sf A})$ and ${\sf (P)}$ hold.}
\KwResult{$\Delta \in \Q[\bm{y}]$ such that $\pi$ realizes a fibration over all
  connected components of $\R^t - \{\Delta = 0\}$}
\SetAlgoNoLine
\For{ all subsets $\bm{i}$ in $\{1, \ldots, q\}$}{
  $\mathcal{M}\gets $ maximal minors of ${\rm jac}([F, G_{\bm{i}}], \bm{x})$ \\
  $\Delta_{\bm{i}} \gets {\sf AlgebraicElimination}([F, G_{\bm{i}},
  \mathcal{M}], \bm{x})$
}
$\Delta \gets \prod_{\bm{i}}\Delta_{\bm{i}}$.\\
\Return $\Delta$.
\caption{{\sf EliminateProper$(F, G, \bm{x}, \bm{y})$}}
\label{algo:elimproper}
\end{algorithm}
\end{center}
\begin{lemma}
  On input $(F, G)$ in $\Q[\bm{x}, \bm{y}]$ satisfying $({\sf A})$, algorithm
  {\sf EliminateProper} is correct.
\end{lemma}
For some applications, deciding if ${\sf (P)}$ holds is easy (e.g. when the
inequalities in $G$ define a box). However, in general, one needs to generalize
{\sf EliminateProper} to situations where ${\sf (P)}$ does not hold.

To do so, we use a classical technique from effective real algebraic geometry.
Let $\varepsilon$ be an infinitesimal and $\R\langle \varepsilon \rangle$ be the
field of Puiseux series in $\varepsilon$ with coefficients in $\R$. By
\cite[Chap. 2]{BPR}, $\R\langle \varepsilon \rangle$ is a real closed field and
one can define semi-algebraic sets over $\R\langle \varepsilon \rangle^{n+t}$.
In particular the set solutions in $\R\langle \varepsilon \rangle^{n+t}$ to the
system defining $S$ is a semi-algebraic set which we denote by ${\rm ext}(S,
\R\langle \varepsilon \rangle)$. We refer to \cite{BPR} for properties of real
Puiseux series fields and semi-algebraic sets defined over such field. 
We make use of the notions of bounded points of $\R\langle
\varepsilon \rangle^n$ over $\R$ (those whose all coordinates have non-negative
valuation) and their limits in $\R$ (when $\varepsilon \to 0$). We denote by
$\lim_0$ the operator taking the limits of such points. 

For $a = (a_1, \ldots, a_n)$, we consider the intersection of ${\rm ext}(S,
\R\langle \varepsilon \rangle)$ with the semi-algebraic set defined by
\[
\Phi^{(a)} = a_1x_1^2+\cdots + a_nx_{n}^2 - 1/\varepsilon\leq 0
\]
where $a_i > 0$ in $\R$ for $1\leq i \leq n$. We denote by $S'_\epsilon$ this
intersection. Since $a_i>0$ for all $1\leq i \leq n$,  $S'_\epsilon$ satisfies ${\sf (P)}$.

\begin{lemma}
  Assume that $(F, G)$ satisfies ${\sf (A)}$. There exists a non-empty Zariski
  open set $\mathscr{A}\subset \C^n$ such that for any choice of $a = (a_1,
  \ldots, a_n)\in \mathscr{A}$, $(F, G^{(a)})$ satisfies $({\sf A})$ with
  $G^{(a)} = G \cup \{\Phi^{(a)}\}$.
\end{lemma}

\begin{proof}
  Let $\bm{i} = \{i_1, \ldots, i_s\}\subset \{1, \ldots, q\}$. We prove below
  that there exists a non-empty Zariski open set $\mathscr{A}_{\bm{i}}\subset
  \C^n$ such that for $(a_1, \ldots, a_n)\in \mathscr{A}_{\bm{i}}$, the
  following property ${\sf (A)}_{\bm{i}}$ holds. Denoting by $G^{(a), \bm{i}}$
  the sequence $(g_{i_1}, \ldots, g_{i_s}, \Phi^{(a)})$, the Jacobian matrix of
  $(F, G^{(a), \bm{i}})$ has maximal rank at any point of $V(F, G^{(a),
    \bm{i}})$. Taking the intersection of the (finitely many)
  $\mathscr{A}_{\bm{i}}$'s is then enough to define $\mathscr{A}$.

  Consider new indeterminates $\alpha_1, \ldots, \alpha_n$ and the polynomial
  $\Phi^{(\alpha)} = \alpha_1 x_1^2+\cdots +\alpha_nx_n^2 - 1/\varepsilon$. Let
  $\Psi$ be the map
  \[
  \Psi:  (x, a)\to F(x), g_{i_1}(x), \ldots, g_{i_s}(x), \Phi^{(a)}(x)
  \]
  Observe that $\mathbf{0}$ is a regular value for $\Psi$ since $(F, G)$
  satisfies ${\sf (A)}$. Hence, Thom's weak
  transversality theorem (see e.g. \cite[App. B]{SaSc17}) implies that there
  exists $\mathscr{A}_{\bm{i}}$ such that ${\sf (A)}_{\bm{i}}$ for any $a\in
  \mathscr{A}_{\bm{i}}$. 
\end{proof}
Assume for the moment that $(F, G')$ satisfies assumption {\sf (A)}. Observe
that the coefficients of $F$ and $G'$ lie in $\Q(\varepsilon)$. Hence, applying
the subroutine {\sf EliminateProper} to $(F, G')$ and the above inequality will
output a polynomial $\Delta_\varepsilon\in \Q(\varepsilon)[\bm{y}]$ such that
the restriction of $\pi$ to $\overline{S'_\varepsilon}$ realizes a
trivialization over each connected component of $\R\langle \varepsilon \rangle^t
- \{\Delta_\varepsilon = 0\}$.
Without loss of generality, one can assume that $\Delta_\varepsilon\in
\Q[\varepsilon][\bm{y}]$ and has content $1$. In other words, one can write
$\Delta_\varepsilon = \Delta_0 + \varepsilon \tilde{\Delta}$ with $\Delta_0\in
\Q[\bm{y}]$ and $\tilde{\Delta}\in \Q[\varepsilon][\bm{y}]$.

\begin{lemma}
  Let $U$ be a connected component of $\R^t-\{\Delta_0 = 0\}$. Then, there
  exists a semi-algebraically connected component $U_\varepsilon$ of $\R\langle
  \varepsilon \rangle^t-\{\Delta_\varepsilon = 0\}$ such that ${\rm ext}(U,
  \R\langle \varepsilon \rangle)\subset U_\varepsilon$.
\end{lemma}

\begin{proof}
  Let $y$ and $y'$ be two distinct points in $U$. Since $U$ is a
  semi-algebraically connected component of $\R^t -\{\Delta_0=0\}$, there exists
  a semi-algebraic continuous function $\gamma: [0, 1]\to U$ with $\gamma(0)=y$
  and $\gamma(1)=y'$ such that $\Delta_0$ is sign invariant over $\gamma([0,
  1])$ (assume, without loss of generality that it is positive). Note also for
  all $t\in [0,1]$, $\Delta_0(\gamma(t))\in \R$. We deduce that
  $\Delta_\varepsilon(\gamma(t)) > 0$ for all $t\in [0,1]$. Now, take
  $\vartheta\in {\rm ext}([0, 1], \R\langle \varepsilon \rangle)$. Observe that
  $\vartheta$ is bounded over $\R$ and then $\lim_0\vartheta$ exists and lies in
  $[0, 1]$. We deduce that $\Delta_{\varepsilon}(\lim_0\vartheta) > 0$ and its
  limit when $\varepsilon\to 0$ is $\Delta_0(\lim_0\vartheta) > 0$ in $\R$. We
  deduce that $\Delta_\varepsilon(\vartheta) > 0$. Hence, $\Delta_\varepsilon$
  is sign invariant over ${\rm ext}(\gamma([0, 1]), \R\langle \varepsilon
  \rangle)$ and then $y$ and $y'$ both lie in the same semi-algebraically
  connected component of $\R\langle \varepsilon \rangle^t - \{\Delta_\varepsilon
  = 0\}$.
\end{proof}
We deduce that there exists $b'\in \N$ such that for all $y\in U$, the number of
semi-algebraically connected components of $S'_\varepsilon\cap\pi^{-1}(y)$ is
$b$. Using the transfer principle as in \cite{BRSS}, we deduce that there exists
$e'\in \R$ positive and small enough such that, the following holds. There
exists $b\in \N$ such that for all $e\in ]0, e'[$ the number of connected
components of $S\cap \{a_1x_1^2+\cdots+a_nx_n^2\leq \frac{1}{e}\}\cap \pi^{-1}(y)$.
is $b$ when $y$ ranges over $U$. This proves the following lemma. 

\begin{lemma}
  Let $U$ be as above. Then the number of connected components of $S_y$ is
  invariant when $y$ ranges over $U$.
\end{lemma}

Finally, we can describe the subroutine {\sf Eliminate} whose correctness
follows from the previous lemma.
\begin{center}
\begin{algorithm}
\KwData{Finite sequences $F$ and $G$ in $\Q[\bm{x}, \bm{y}]$ with $\bm{x} =
  (x_1, \ldots, x_n)$ and $\bm{y} = (y_1, \ldots, y_t)$, defining a
  semi-algebraic set $S\subset \R^n \times \R^t$. \\
Assumes that $(F, G)$ satisfies assumption $({\sf A})$.}
\KwResult{$\Delta \in \Q[\bm{y}]$ such that the number of connected components
  of $S_y$ is invariant when $y$ ranges over a connected component of 
  $\R^t - \{\Delta = 0\}$}
\SetAlgoNoLine
{Choose $a_1>0, \ldots, a_n>0$ in $\Q$ randomly and let $g \gets a_1
  x_1^2+\cdots+a_n x_n^2 \leq \frac{1}{\varepsilon}$} \\
{$\Delta \gets {\sf EliminateProper}(F, G \cup g, \bm{x}, \bm{y})$}\\
{$\Delta \gets {\sf Normalize}(\Delta)$}
\Return $\Delta_0$.
\caption{{\sf Eliminate$(F, G, \bm{x}, \bm{y})$}}
\label{algo:elim}
\end{algorithm}
\end{center}
\section{Computations}\label{sec:experiments}
We have implemented several variants of the roadmap algorithms sketched in
Section~\ref{sec:roadmap} as well as variants of the algorithm {\sf
  ParametricSolve}. To perform algebraic
elimination, we use Gr\"obner bases implemented in the {\sc FGb} library by
J.-C. Faug\`ere \cite{FGb}. The roadmap algorithm and the routines for
computing sample points in semi-algebraic sets are implemented
in the {\sc RAGlib} library \cite{RAGlib}.

We have not directly applied the most general version of {\sf ParametricSolve}
to the polynomial $B$. Indeed, since its variables $v_2, v_3, v_4$ lie in the
Cartesian product $\mathbb{P}^1(\R)\times\mathbb{P}^1(\R)\times\mathbb{P}^1(\R)$
(which is compact), the projection on the parameter's space is proper and it
suffices to compute critical loci of that projection. There is one technical
(but easy) difficulty to overcome: polynomial $B$ actually admits a positive
dimensional singular locus. But an easy computation shows that this singular
locus has one purely complex component (which satisfies $v_4^2+1$) which can then
be forgotten. The other component has a projection on the paramaters'space which
Zariski closed (it is contained in the set satisfied by $a_2a_3=0$). This way,
we directly obtain the following polynomial for $\Delta$ by computing the
critical locus and consider additionally the set defined by $a_2a_3=0$. {\small\[
    {{a_2}}{{a_3}}{d_5}\, \left( {a_2}+{a_3}+{d_5} \right) \left( {a_2}+{a_3}-
      {d_5} \right)
  \]} Computing $\Delta$ as above does not take more than $3$ sec. on a standard
laptop using {\sc FGb}. Getting sample points in the set defined by 
$\Delta\neq 0$ is trivial. We obtain the following $10$ sample points using {\sc RAGlib}
{\tiny
  \begin{dmath*}
\{{a_2}=-1,{a_3}=-3,{d_5}=3\}, \{{a_2}=-1,{a_3}=-1,{d_5}=3\}, \{{a_2}=-1, {a_3}=2, {d_5}=3\}, \{{a_2}=-1,{a_3}=5,{d_5}=3\}, \{{a_2}=-1,{a_3}=\frac{1}{2},{d_5}=3\}, \{{a_2}=1,{a_3}=-120,{d_5}=118\}, \{{a_2}=1,{a_3}=-118,{d_5}=118\}, \{{a_2}=1,{a_3}=1,{d_5}=118\}, \{{a_2}=1,{a_3}=118,{d_5}=118\}, \{{a_2}=1,{a_3}=-1/2,{d_5}=118\}
  \end{dmath*}
}Our implementation allows us to compute a roadmap for one sample
point within $20$ minutes on a standard laptop. Analyzing the connectivity of
these roadmaps is longer as it takes $40$ min. All in all, approximately $10$
hours are required to handle this positive dimensional parametric system. The
data we computed are available at {\tt
  http://ecarp.lip6.fr/papers/materials/issac20/}. These computations allow to
retrieve the conclusions of our theoretical analysis of the UR family. They
illustrate that prototype implementations of our algorithms are becoming
efficient enough to tackle automated kinematic singularity analysis in robotics.


\paragraph*{Acknowledgments.}%
The three authors are supported by the joint ANR-FWF ANR-19-CE48-0015, FWF I
4452-N \textsc{ECARP} project. Mohab Safey El Din is supported by the ANR grants
ANR-18-CE33-0011 \textsc{Sesame} and ANR-19-CE40-0018 \textsc{De Rerum Natura},
the PGMO grant \textsc{CAMiSAdo} and the European Union's Horizon 2020 research
and innovation programme under the Marie Sklodowska-Curie grant agreement N.
813211 (POEMA).

\end{document}